\documentclass[a4paper,10pt]{amsart}
\usepackage[english]{babel}
\usepackage{amsmath,amsthm,amssymb}
\usepackage{verbatim,graphicx,hyperref}
\usepackage[T1,T5]{fontenc}
\usepackage[retainorgcmds]{IEEEtrantools}
\newtheorem{theorem}{Theorem}[section]
\newtheorem{proposition}[theorem]{Proposition}
\newtheorem{lemma}[theorem]{Lemma}

\theoremstyle{definition}

\newtheorem{example}[theorem]{Example}

\theoremstyle{remark}
\newtheorem*{remark}{Remark}

\numberwithin{equation}{section}

\newcommand{\pqbin}[2]{\genfrac{(}{)}{0pt}{0}{#1}{#2}_{\!\!p,q}}
\newcommand{\pqbint}[2]{\genfrac{(}{)}{0pt}{2}{#1}{#2}_{p,q}}
\newcommand{\dd}{\,\mathrm{d}}

\begin{document}

\author{P. {\AA}hag} \address{Department
  of Mathematics and Mathematical Statistics, Ume{\aa} University,
  SE-901 87 Ume{\aa}, Sweden}
 \email{per.ahag@umu.se}

\author{R. Czy{\.z}}\address{Faculty
  of Mathematics and Computer Science, Jagiellonian University, \L
  ojasiewicza 6, 30-348 Krak\'ow, Poland}
 \email{rafal.czyz@im.uj.edu.pl}

\author{P. H. Lundow}
\address{Department of Mathematics and Mathematical Statistics,
  Ume{\aa} University, SE-901 87 Ume{\aa}, Sweden}
\email{per.hakan.lundow@math.umu.se}


\title[On a generalised Lambert $W$ branch transition function]{On a
  generalised Lambert $W$ branch transition function arising from
  $p,q$-binomial coefficients} \date{\today}

\begin{abstract}
  With only a complete solution in dimension one and partially solved
  in dimension two, the Lenz-Ising model of magnetism is one of the
  most studied models in theoretical physics. An approach to solving
  this model in the high-dimensional case ($d>4$) is by modelling the
  magnetisation distribution with $p,q$-binomial coefficients. The
  connection between the parameters $p,q$ and the distribution peaks
  is obtained with a transition function $\omega$ which generalises
  the mapping of Lambert $W$ function branches $W_0$ and $W_{-1}$ to
  each other.  We give explicit formulas for the branches for special
  cases. Furthermore, we find derivatives, integrals,
  parametrizations, series expansions, and asymptotic behaviors.
\end{abstract}

\maketitle

%==============================================================================
\section{Introduction and background}\label{section1}

We will study the transition function $\omega$ that arises when using
$p,q$-binomial coefficients to attack the model of magnetism in
statistical mechanics for higher dimensions. To get an overall view of
our study, we start in Section~\ref{section1_PO} with a physical
background on the Lenz-Ising model. In Section~\ref{section1_PQcoeff}
we give an introduction to $p,q$-binomial coefficients, and how it
relates to the given physical problem. Section~\ref{section1_PQcoeff}
will also reveal how the transition function arises and it's relation
to the Lambert $W$ function.

\subsection{Physical origins}\label{section1_PO}
The Lenz-Ising model, usually referred to as the Ising model, was
introduced by Wilhelm Lenz~\cite{Lenz:1920} in 1920 as a simple model
for magnetism. The time was ripe for such a model after the
discoveries by Pierre Curie\footnote{In his thesis 1895. He and his
wife Marie later shared the 1903 Nobel prize in physics with Becquerel
for their work on radioactivity.} that magnetic materials (iron,
cobalt, nickel, etc.) undergo a phase transition at a critical
temperature $T_c$, the Curie temperature, above which they lose their
permanent magnetic properties. In the Ising model without an external
field, the material is described as a system of interacting particles,
each having spin $\pm 1$, and is governed by the
temperature~\cite{Cipra:1987,DuminilCopin:2022,Okounkov:2022}. The
system in question can be any finite graph, with the vertices
corresponding to the particles and the edges indicating which
particles interact. In the model's most famous version, the system is
the (infinite) $d$-dimensional integer lattice. The $1$-dimensional
model (an infinite chain of particles) was solved in 1925 by Lenz's
student, Ernst Ising, in his thesis~\cite{Ising:1925}. Unfortunately,
the result was disappointing since he showed that this system does not
have a phase transition at any positive temperature, the critical
temperature being $T_c=0$.

It took until 1944 when the Norwegian chemist Lars
Onsager\footnote{Nobel prize in chemistry 1968 for his work on the
thermodynamics of irreversible processes.} solved the model for
2-dimensional systems~\cite{Onsager:1944}. This solution was
considered a major breakthrough, rendering a positive critical
temperature, $T_c=2/\log(1+\sqrt{2})$ (assuming unit interaction
between nearest-neighbor particles in the lattice). Unfortunately,
again, the techniques used by Onsager did not point to a solution for
the 3-dimensional model. In fact, to this day, very little is known
rigorously about the 3-dimensional model, and the 2-dimensional model
with an external field is still unsolved. For dimensions $d\ge 4$ the
critical exponents, which govern the behaviour of various quantities
near the critical point, are known exactly but the critical
temperature is still not known exactly for any $d>2$.

This has not made the Ising model any less attractive, instead, it has
generated a staggering number of papers studying many variants of the
model in different
dimensions~\cite{Stanley:1968,Baxter:1973,SherringtonKirkpatrick:1975,EdwardsAnderson:1975}. Also
it has become a testing bed for various Monte Carlo simulation
algorithms~\cite{Wolff:2001,WangLandau:2001,FerrenbergXuLandau:2018}.

Let us focus a little closer on an interesting aspect of the
Ising model without an external field, namely its magnetization distribution. 
The magnetization $M$ of the Ising model is the sum of all the spins $S_i=\pm 1$ in the system,
$M=\sum_iS_i$, so that for a system with $n$ spins we have $-n\le M\le
n$.  For a finite $d$-dimensional system (we assume $d\ge 2$), at any
given temperature $T$, the magnetization $M$ obeys a symmetric
distribution with peaks at $\pm a n$, for some $0\le a\le 1$. Now, for
temperatures above the Curie temperature ($T>T_c$), the distribution
of magnetizations is unimodal with its peak at $M=0$ so that
$a=0$. When we lower the temperature below the Curie temperature
($T<T_c$) the distribution becomes bimodal with peaks at $\pm a n$ for
$0<a\le 1$, with $a\to 1$ as $T\to 0$. The parameter $a$ here
corresponds to the so-called spontaneous magnetization of the Ising
model, but its relation to $T$ is only known exactly for
$2$-dimensional (infinite) systems. It was conjectured by Onsager in
1948 and proved by Yang in 1952 (see~\cite{McCoyWu:1973} for
details) that in this case
\[a(T) = \begin{cases}
  (1 - \sinh(2/T)^{-4})^{1/8}&\text{if } T<T_c\\
  0&\text{if } T\ge T_c
\end{cases}
\]
where $T_c=2/\log(1+\sqrt{2})$.  Thus a phase transition occurs at
$T=T_c$ when $a$ becomes positive. In physical terms, the unimodal
distribution corresponds to losing the magnetic properties, whereas
the bimodal distribution corresponds to retaining them. Finding a
corresponding formula for dimensions $d\ge 3$ would be a major
breakthrough.

\subsection{The $p,q$-binomial coefficients}\label{section1_PQcoeff}
It was suggested in~\cite{LundowRosengren:2010} that the magnetization
distribution is well described by $p,q$-binomial coefficients for
finite high-dimensional systems ($d\ge 5$). In fact, in a special case
(mean-field), they are equivalent, and for $d\ge 5$ they have the same
asymptotic shape when $n\to\infty$. In principle, one could then model
the magnetization distribution for a finite system at temperature $T$
with a $p,q$-binomial distribution where $p$ and $q$ depend on $T$
according to some functions $p(T)$ and $q(T)$.  It would be a more
realistic project to find these functions for temperatures very close
to $T_c$ and one such attempt was made in~\cite{LundowRosengren:2013}.

However, it should be mentioned that the $p,q$-binomial coefficients
have received attention also for other interesting purely mathematical
properties~\cite{AhmiaBelbachir:2018,AhmiaBelbachir:2019,Corcino:2015}.

Let us here provide some more detail on these coefficients. The
$p,q$-binomial coefficient~\cite{Corcino:2008} is defined for $p\neq
q$ as
\begin{equation*}\label{pqbin}
  \pqbin{n}{k} = \prod_{j=1}^k\frac{p^{n-k+j} - q^{n-k+j}}{p^j - q^j},\quad 0\le k\le n
\end{equation*}
from which a symmetric $p,q$-binomial
distribution~\cite{LundowRosengren:2010,LundowRosengren:2013} is
defined to have a probability mass function proportional to the
sequence of these coefficients. Note that the sum of the
$p,q$-binomial coefficients do not seem to have a simple expression,
as opposed to standard binomial coefficients for which the sum is
simply $2^n$. For $p,q>0$ it has been shown~\cite{SuWang:2012} that
the coefficients either form a unimodal sequence with maximum at
$k=\lfloor n/2\rfloor$, or, a bimodal sequence with maxima at $k$ and
$n-k$ for some $0\le k\le n/2$.  We let the parameter $a$ control the
location of the sequence maximum by defining
\begin{equation*}
  k = \frac{n}{2}(1-a),\quad 0 \le a \le 1.
\end{equation*}
Having two consecutive $p,q$-binomial coefficients, indexed $k-1$ and
$k$, being equal leads, after simplification, to the equation
\begin{equation}\label{eqratio1}
  \frac{p^{n-k+1} - p^k}{q^{n-k+1} - q^k} = 1,\quad 1\le k\le \lfloor n/2\rfloor.
\end{equation}
We here need to introduce the $p,q$-parameterizations
\begin{align*}
  &p = 1 + \frac{2y}{n}\\
  &q = 1 + \frac{2z}{n},\quad z<0
\end{align*}
though other parameterizations are also of interest.  With $a$ and $z$
fixed we now solve for $y$.

\subsection{Introducing the functions $\omega$ and $\psi$}\label{section1_fcn}
First the special case $a=0$. The
asymptotic form of the ratio in~\eqref{eqratio1} now becomes
simply
\begin{equation*}\label{eqratio2}
  \frac{y e^y}{z e^z} + \mathcal{O}(1/n) = 1
\end{equation*}
and to receive a leading term of $1$ we must have $ze^z=ye^y$. This
defines a special case of the transition function $\omega$:
\begin{equation}\label{eqomega0}
  y = \omega(0,z) = \begin{cases}
    W_0(ze^z)&\text{if } z<-1\\
    W_{-1}(ze^z)&\text{ if } -1\le z<0\\
  \end{cases}.
\end{equation}
Here $W$ denotes the famous Lambert $W$ function, which returns one of
two real solutions $w=W_i(x)$ of the equation $x=w e^w$. The principal
solution, $w=W_0(x)$ defined for $x\ge -1/e$, gives $w\ge -1$ and the
other branch, $w=W_{-1}(x)$ defined for $-1/e \le x < 0$, gives $w<-1$.
The {\it transition function} $\omega(0,z)$ thus maps solutions
between the two branches of the Lambert $W$ function. The coefficient
sequence then will have its two maxima at points $k=n/2 \pm \ell$
where $\ell = o(n)$. In Fig.~\ref{fig0} the right panel shows an
example of this case and the left panel shows the case $a=1/2$.

\begin{figure}
  \begin{minipage}{0.45\textwidth}
    \begin{center}
      \includegraphics[width=\textwidth]{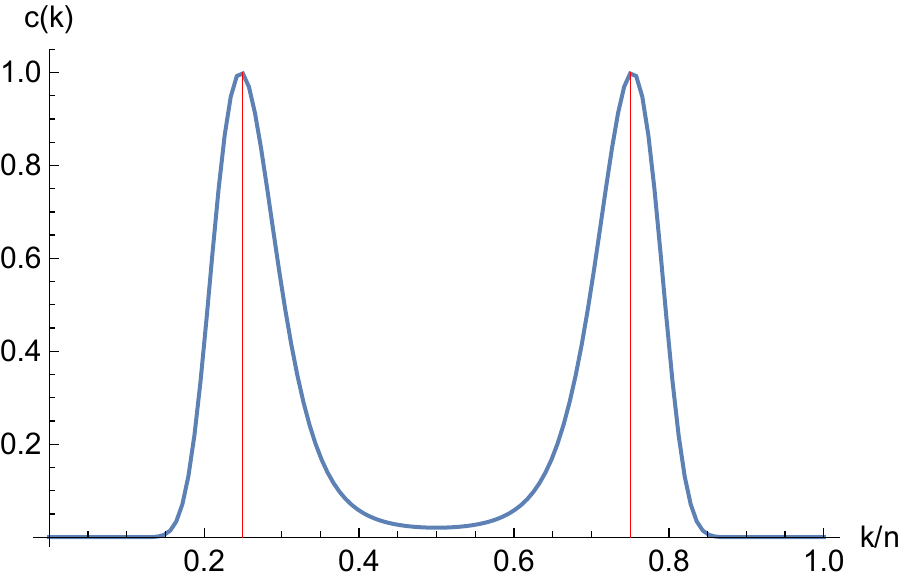}
    \end{center}
  \end{minipage}%
  \begin{minipage}{0.45\textwidth}
    \begin{center}
      \includegraphics[width=\textwidth]{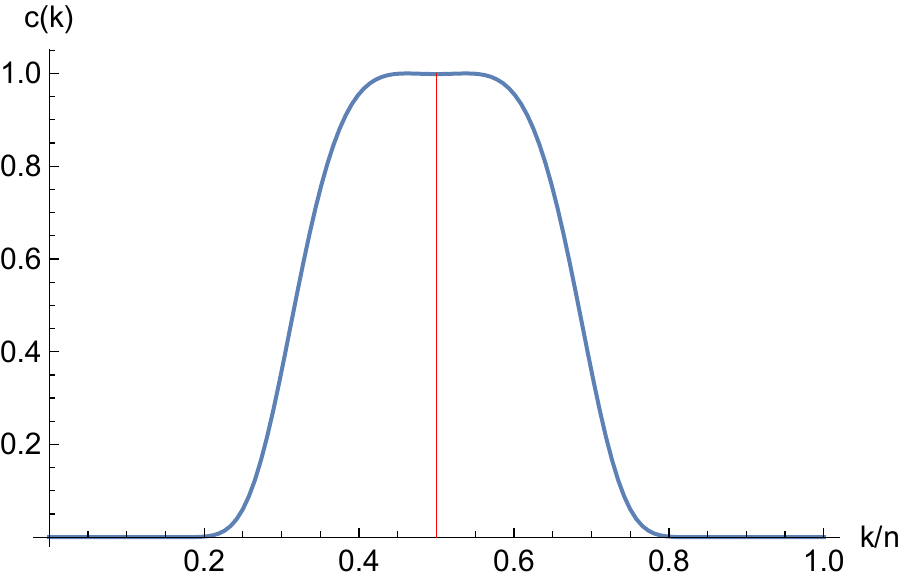}
    \end{center}
  \end{minipage}
  \caption{\label{fig0} Plots of $c(k)=\tfrac{1}{C}\pqbint{n}{k}$ vs
    $k/n$ for $n=2^{10}$, $z=-5$ and $y=\omega(a,z)$, where $C$ is a
    normalizing factor. Red vertical lines at $k/n=(1\pm a)/2$.
    Panels show the cases $a=1/2$ with $y=-0.0891004$ (left) and $a=0$
    with $y=-0.0348858$ (right).  As $n\to\infty$ the peak locations
    will approach the position of the vertical lines. Also, the
    coefficient mass will become more concentrated around these
    lines.}
\end{figure}

We will now generalise the transition function of $W$ to the case of
$0<a<1$. The ratio of~\eqref{eqratio1} now becomes
\begin{equation*}\label{eqratio3}
  \frac{\sinh(ay)e^y}{\sinh(az)e^z} + \mathcal{O}(1/n) = 1
\end{equation*}
so to receive the correct leading term of $1$ we must now solve
\begin{equation}\label{eqbranches}
  \sinh(ay)e^y = \sinh(az)e^z.
\end{equation}
We therefore define the function
\begin{equation*}
  f(a,w) = \sinh(aw)e^w
\end{equation*}
and introduce its inverse function $w=\psi_i(a,x)$ as one of the two
real solutions of
\begin{equation*}
  x=f(a,w),\quad 0 < a < 1.
\end{equation*}
We note that $f(a,w)$ takes its minimum value $L_a$ at $w=M_a$, where
\begin{align*}
  L_a &= \frac{-a}{\sqrt{1-a^2}}\left(\frac{1-a}{1+a}\right)^\frac{1}{2a},\\
  M_a &= \frac{1}{2a}\log\left(\frac{1-a}{1+a}\right).
\end{align*}
The principal branch $\psi_0(a,x)\ge M_a$ is now defined for $x\ge
L_a$, while the other branch $\psi_{-1}(a,x)\le M_a$ is defined for
$L_a\le x<0$.  Thus $M_a$ is the branch separator with $f(a,M_a)=L_a$
and $\psi_0(a,L_a)=\psi_{-1}(a,L_a)=M_a$. In Fig.~\ref{fig1} we show
how the equation $f(a,w)=x$ relates to the definition $\psi$ and
$\omega$.

\begin{figure}
  \begin{center}
    \includegraphics[width=0.9\textwidth]{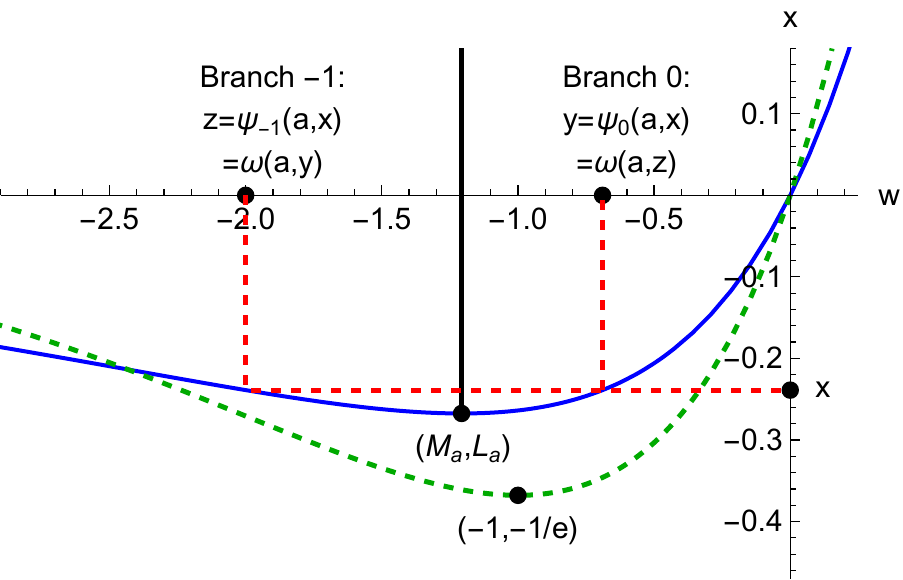}
  \end{center}
  \caption{\label{fig1} Plots of $we^w$ (dashed green) and
    $\sinh(aw)e^w$ (solid blue) for $a=2/3$. A point $z$ in one branch
    is mapped to a point $y=\omega(a,z)$ in the other branch when
    $z<0$, the two branches separated by $x=M_a$. Note
    $x=\sinh(ay)e^y=\sinh(az)e^z$ where $L_a\le x < 0$.}
\end{figure}

The transition function $\omega$, defined above only for $a=0$
in~\eqref{eqomega0}, can now be generalized to $0<a<1$
\begin{equation}\label{eqomega1}
  y = \omega(a,z) = \begin{cases}
    \psi_0(a,f(a,z))&\text{if } z<M_a\\
    \psi_{-1}(a,f(a,z))&\text{ if } M_a\le z<0\\
  \end{cases}.
\end{equation}
Plots of $\psi$ and $\omega$ are shown in Fig.~\ref{fig23}.

\begin{figure}
  \begin{minipage}{0.45\textwidth}
    \begin{center}
      \includegraphics[width=\textwidth]{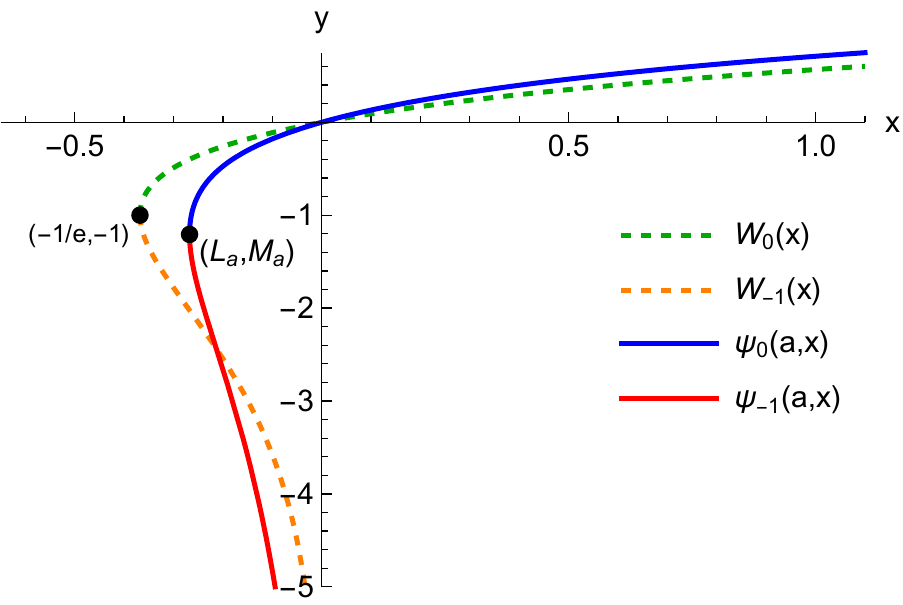}
    \end{center}
  \end{minipage}%
  \begin{minipage}{0.45\textwidth}
    \begin{center}
      \includegraphics[width=\textwidth]{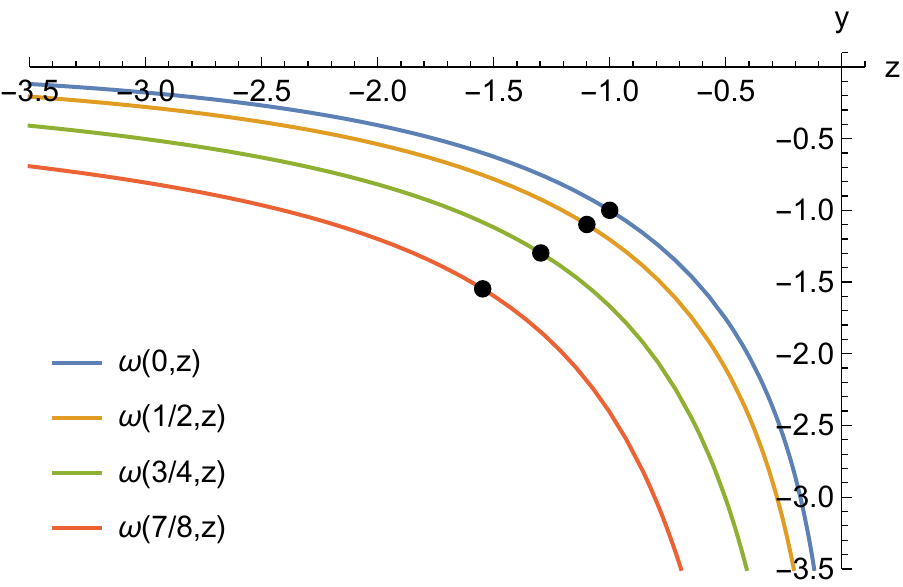}
    \end{center}
  \end{minipage}
  \caption{\label{fig23} Left: Plots of branch $0$ (above the points)
    and branch $-1$ (below the points) of $y=W(x)$ (dashed, green and
    orange) and $y=\psi(2/3,x)$ (solid, blue and red). Right: Plots of
    $\omega(a,z)$ for $a=0$, $1/2$, $3/4$, $7/8$ (downwards). Points
    at $(M_a,M_a)$. }
\end{figure}

Note that $\lim_{a\to 0^+} \omega(a,z)=\omega(0,z)$ so that $\omega$
for $0<a<1$ is indeed a natural and continuous generalisation of the
transition function for Lambert $W$ of~\eqref{eqomega0}.
However, the function $\psi$ is only related to the Lambert $W$ for small
$a$. In that case we have $f(a,w)\sim awe^w$ so that
\begin{equation*}
  \psi(a,x)\sim W(x/a),\quad x> L_a
\end{equation*}
for both branches, where the functions are defined.
Let us also briefly state the following properties of $\omega$.
\begin{enumerate}\itemsep2mm
\item $z=M_a$ is a fixed point, i.e., $\omega(a,M_a)=M_a$.
\item $\omega$ is an involution, i.e., $\omega(a,\omega(a,z)) = z$ for all $z<0$.
\item $\omega(a,z)$ is negative, decreasing and concave for all $z<0$.
\item $\lim_{z\to -\infty}\omega(a,z)=0$ and $\lim_{z\to 0^-}\omega(a,z)=-\infty$.
\end{enumerate}

The transition function $\omega$ was used in~\cite{LundowRosengren:2010}, though in a slightly different
guise, and very little was mentioned with regard to its mathematical
properties. This article aims to remedy this by first investigating
the properties of $\psi$ and apply this to $\omega$.

To conclude this section, the $p,q$-binomial coefficient sequence will
then have its peaks at $k=n/2(1\pm a) \pm \ell$ where $\ell=o(n)$ and
the parameter $z$ controls the mass shape of the sequence, see
Fig.~\ref{fig0}. Note that the peaks are never at exactly $k=n/2(1\pm
a)$.  To obtain this the resulting $y$ would also depend on $n$ and a
special function $\bar\omega(n,a,z)$ is needed. We would then have
$\bar\omega(n,a,z)\to \omega(a,z)$ when $n\to\infty$.

%==============================================================================
\section{Results on the $\psi$- and $\omega$-functions from elementary techniques}\label{section3}

In this section we will show some basic properties of the $\psi$ and
$\omega$-function obtained only from elementary techniques. Also we
solve some interesting special cases. Later in the article we will use
more powerful methods and extend these results considerably.

%------------------------------------------------------------------------------
\subsection{The boundary cases $a=0$ and $a=1$ and a lower bound for $\psi_0$}

First we recall that
\begin{equation}
  f(a,w) = \sinh(aw)e^w = \frac{1}{2}\left(e^{w(1+a)} -
  e^{w(1-a)}\right), \quad 0\le a\le 1,\,w\in\mathbb{R}.
\end{equation}
The limits for $a\to 0^+$ and $a\to 1^-$ when $w$ is fixed are then
\begin{align*}
  &f(0,w) = \lim_{a\to 0^+} f(a,w) =  0,\\
  &f(1,w) = \lim_{a\to 1^-} f(a,w) = \frac{1}{2}\left(e^{2w}-1\right).
\end{align*}
For fixed $w$ we also have:
\begin{enumerate}
  \itemsep 1ex
\item If $w$ is positive then $f(a,w)$ is increasing for $0<a<1$.
\item If $w$ is negative then $f(a,w)$ is decreasing for $0<a<1$.
\end{enumerate}

Now we can define the limit forms of $\psi$.  Solving for $w$ in
$f(1,w)=x$ gives the limit inverse function
\begin{equation*}
  \psi_0(1,x) = \lim_{a\to 1^-}\psi_0(a,x) = \frac{1}{2}\log(1+2x),\quad x>-\frac{1}{2},
\end{equation*}
thanks to uniform convergence. In fact, since $f(a,w)$ is increasing
with respect to $a$ then $\psi_0(a,x)$ is decreasing with $a$ so that
we receive the inequality
\begin{equation*}
  \psi_0(a,x) \ge \frac{1}{2}\log(1+2x),\quad 0<a<1,\,x>-\frac{1}{2}.
\end{equation*}
In Theorem~\ref{thmpsi0} below we will give a much sharper bound, but
requiring considerably more work. Since $f(1,w)>-1/2$ is strictly
increasing for all $w$, thus having no second branch, we obtain no
limit $\psi_{-1}(1,x)$.  Also, since $f(0,w)=0$ we have no
well-defined inverse function $\psi_{-1}(0,x)$.

%------------------------------------------------------------------------------
\subsection{A parametrization of $\psi_0$ for general $a$}

In Section~\ref{section5} we will use a powerful technique which gives
a parametrization of both branches $\psi_0$ and $\psi_{-1}$ for
$L_a\le x<0$. However, here we only require an elementary technique
which allows us to develop a parametrization of $\psi_0(a,x)$ for
$x>0$. First we define
\[
u(x) = \tfrac{1}{2}x(x^{a}-x^{-a}) = \tfrac{1}{2}(x^{1+a}-x^{1-a}),
\quad x>0.
\]
Then $u\circ \exp(x) = \sinh(ax)e^x = f(a,x)$, and therefore
$\log\circ u^{-1}(x) = \psi(a,x)$. Hence, $e^{\psi(a,x)}=u^{-1}(x)$ so
that $e^{\psi(a,u(x))}=x$. Finally, we now require $x>0$, thus giving
the branch $\psi_0$, so that we can obtain
\[
\psi_0\!\left(a,\tfrac{1}{2}(x^{1+a}-x^{1-a})\right)=\log x,
\]
which gives the sought-after parametrization, defined for $\beta>1$:
\begin{align*}
  &x =\tfrac{1}{2}(\beta^{1+a}-\beta^{1-a}),\\
  &\psi_0(a,x)=\log\beta.
\end{align*}

%------------------------------------------------------------------------------
\subsection{Special values of $f$ and $\psi_{-1}$ for general $a$}

The $\psi_{-1}$-branch is more difficult to parametrize for general
$a$.  However, it is easy to give an infinite sequence of points
$(x_n,y_n)$ satisfying $\psi_{-1}(a,x_n)=y_n$ and we will do so here.

Note that it is an easy exercise to show that the $n$th derivative of
$f$ vanishes at $n$th multiples of $M_a$, so that
\begin{equation}
  \frac{\partial^n}{\partial w^n}f(a,w) = 0\text{ for }w =
  n\,M_a,\quad n=1,2,\ldots
\end{equation}
If we define
\begin{equation}
  I_a(n) = \frac{1}{2}
  \left(\frac{1-a}{1+a}\right)^{\frac{n(1-a)}{2a}}\left(\left(\frac{1-a}{1+a}\right)^n-1\right)
\end{equation}
so that $I_a(1)=L_a$ and $I_a(2)=-2L_a^2/a$ is the inflection point of
$\psi_{-1}$ etc, then
\begin{equation}
  f(a,nM_a) = I_a(n).
\end{equation}
Of course, since $w=nM_a\le M_a$ for $n\ge 1$, this only gives the
solutions in branch $-1$ of the equation $f(a,w)=I_a(n)$, thus leading
to
\begin{equation}
  \psi_{-1}(a,I_a(n)) = n\,M_a,\quad n=1,2,\ldots
\end{equation}
Since there is also a corresponding solution for $\psi_0(a,I_a(n))$
but this value does not seem to have a simple expression. However,
for special values of $a$ we can find closed form
expressions for $\psi$ but we will return to this in Section~\ref{section5}.

%------------------------------------------------------------------------------
\subsection{The special case $a=1/3$: explicit formulas of $\psi$ and $\omega$}

In general it is difficult to find explicit formulas for
$\psi$ and $\omega$. However, for some special cases of $a$ we will be
successful and the case $a=1/3$ turns out to be surprisingly easy to
handle. In Section~\ref{section5} we will give formulas for $\psi$ and
$\omega$ at other rational values of $a$ but they are considerably
more complicated since they are based on the roots of the quartic
polynomials. The formulas given here for $\psi$ and $\omega$ seem to
be the only simple exact solutions.

To compute $\psi(a,x)$ we need to solve for $w$ in the equation
$f(a,w)=x$.  When $a=1/3$ we then receive
\begin{equation}
  f(\tfrac{1}{3},w)
  = \frac{1}{2} \left(e^{\frac{4w}{3}} - e^{\frac{2w}{3}}\right)
  = \frac{1}{2} \left(Y^2 - Y\right) = x
\end{equation}
where
\[Y = e^{\frac{2w}{3}}\]
Solving the second degree polynomial gives the roots
\begin{equation}
  Y = \frac{1}{2} \pm \sqrt{2x+\tfrac{1}{4}}
\end{equation}
where the two solutions correspond to the two branches so that
\begin{align}
  \psi_0(\tfrac{1}{3},x) &= \frac{3}{2}\log\left(\frac{1}{2}
  + \sqrt{2x+\frac{1}{4}}\right), \quad -\frac{1}{8}\le x \\
  \psi_{-1}(\tfrac{1}{3},x) &= \frac{3}{2}\log\left(\frac{1}{2}
  - \sqrt{2x+\frac{1}{4}}\right), \quad -\frac{1}{8}\le x < 0.
\end{align}

We now combine this with the definition of $\omega$ in~\eqref{eqomega1}.  The case $w<M_{1/3}=-\tfrac{3}{2}\log 2$
corresponds to $0<Y<1/2$ and we receive
\begin{align*}
  \omega(\tfrac{1}{3},w) &= \psi_0(\tfrac{1}{3},f(\tfrac{1}{3},w))
  = \frac{3}{2}\log\left(\frac{1}{2} + \sqrt{Y^2-Y+\frac{1}{4}}\right) \\
  &= \frac{3}{2}\log\left(\frac{1}{2} + \left|Y-\frac{1}{2}\right|\right)
  = \frac{3}{2}\log(1-Y)
  = \frac{3}{2}\log\left(1 - e^{\frac{2w}{3}}\right).
\end{align*}
The other case $ -\tfrac{3}{2}\log 2 \le w < 0$ gives $1/2\le Y < 1$
so that
\begin{align*}
  \omega(\tfrac{1}{3},w) &= \psi_{-1}(\tfrac{1}{3},f(\tfrac{1}{3},w))
  = \frac{3}{2}\log\left(\frac{1}{2} - \sqrt{Y^2-Y+\frac{1}{4}}\right) \\
  &= \frac{3}{2}\log\left(\frac{1}{2} - \left|Y-\frac{1}{2}\right|\right)
  = \frac{3}{2}\log(1-Y)
  = \frac{3}{2}\log\left(1 - e^{\frac{2w}{3}}\right).
\end{align*}
We thus receive the same formula in both cases so we conclude that
\begin{equation}
  \omega(\tfrac{1}{3}, z) =
  \frac{3}{2}\log\left(1-e^{\frac{2z}{3}}\right), \quad z < 0.
\end{equation}

It is now easy to give, for example, the definite integral of $\omega$,
\begin{equation}
  \int_{-\infty}^0 \omega(\tfrac{1}{3},z)\dd z = -\frac{3\pi^2}{8},
\end{equation}
but in Section~\ref{section6} we will give the general formula for $a$.
now easily obtained in this case
\begin{align*}
  \psi_0(\tfrac{1}{3},x) &\sim \frac{3}{4}\log(2x) + \frac{3}{4\sqrt{2x}},
  \quad x\to\infty\\
  \psi_{-1}(\tfrac{1}{3},x) &\sim \frac{3}{2}\log(-2x) -3x + 9x^2,
  \quad x\to 0^-
\end{align*}
and we will state general asymptotes in Section~\ref{section9} and
\ref{section10}.

%------------------------------------------------------------------------------
\subsection{$\psi(a,x)$ is transcendental number}
We shall prove that if $x\neq 0$ is algebraic number ($x\in \mathcal
A$), then $\psi(a,x)$ is transcendental number ($\psi(a,x)\in \mathcal
T$).  To prove it we need Lindemann-Weierstrass Theorem
(~\cite[Theorem 1.4]{Baker:1975}), i.e., if $a_1,\dots,a_n\in \mathcal
A$ and $\alpha_1,\dots,\alpha_n\in \mathcal A$ are distinct numbers,
then the equation
\[
a_1e^{\alpha_1}+\dots+a_ne^{\alpha_n}=0
\]
has only the trivial solution $a_1=\dots=a_n=0$.

\begin{proposition}
  If $a, x\in \mathcal A$, $x\neq 0$, then $\psi(a,x)\in \mathcal T$.
\end{proposition}

\begin{proof}
  Assume that $a, x\in \mathcal A$, $x\neq 0$ and $\psi(a,x)\in \mathcal
  A$. Then $1+a, 1-a, 2x\in \mathcal A$ and therefore
  $(1+a)\psi(a,x),(1-a)\psi(a,x)\in \mathcal A$. By the definition of
  $\psi$ we have
  \[
  e^{(1+a)\psi(a,x)}-e^{(1-a)\psi(a,x)}-2xe^0=0,
  \]
  which is impossible by the Lindemann-Weierstrass Theorem.
\end{proof}

%==============================================================================
\section{Derivatives and a primitive function of $\psi$}\label{section4}
In Proposition~\ref{derivative}, we provide a formula for the $n$th
derivative of $\psi$ with respect to $x$. Here, a shorter notation
$\psi(x)=\psi(a,x)$, is used referring to both branches, and $a$ is
treated as a fixed parameter. A primitive function of $\psi$ is
obtained later in Proposition~\ref{ppsi}.

\begin{proposition}\label{derivative}
  For $n\in \mathbb N$
  \begin{equation}\label{ndpsi}
    \psi^{(n)}(x)=\frac{P_n(\cosh(a\psi(x)),\sinh(a\psi(x)))e^{-n\psi(x)}}
        {\left(a\cosh(a\psi(x))+\sinh(a\psi(x))\right)^{2n-1}},
  \end{equation}
  where the polynomials $P_n$ are given by the following recurrence formula:
  $P_1(x,y)=1$ and
  \begin{align*}
    &P_{n+1}(x,y)=P_n(x,y)((a-3na)x+(a^2-n-2na^2)y)\\
    &+\frac{\partial P_n}{\partial x}(x,y)(a^2xy+ay^2) +
    \frac{\partial P_n}{\partial y}(x,y)(axy+a^2x^2).
  \end{align*}
  
\end{proposition}

\begin{proof}
  We proceed by induction. We shall use the shorter notation
  $\psi(x)=\psi(a,x)$ when referring to both branches. The case $n=1$,
  can be deduced by using the implicit function theorem:
  \begin{equation*}
    \psi'(x)
    = \frac{e^{-\psi(x)}}{\left(a\cosh(a\psi(x))+\sinh(a\psi(x))\right)}
    = \frac{1}{x\left(a\coth(a\psi(x))+1\right)}
  \end{equation*}
  with the special value
  \[
  \psi_0'(0) = \frac{1}{a}.
  \]
  Assume next that~\eqref{ndpsi} is valid for some natural number
  $n$. Let us define
  \begin{align*}
    g(x) &= a\cosh(a\psi(x))+\sinh(a\psi(x))\\
    X & =\cosh(a\psi(x))\\
    Y &= \sinh(a\psi(x))
   \end{align*}
  We then have
  \begin{align*}
    \psi^{(n+1)} &=\frac{e^{-n\psi}\psi'\left(-nP_n + \frac{\partial
        P_n}{\partial x}a\sinh(a\psi) + \frac{\partial P_n}{\partial
        y}a\cosh(a\psi)\right)g^{2n-1}}{g^{4n-2}}\\
    & -\frac{e^{-n\psi}\psi'P_n\left(a^2\sinh(a\psi) +
      a\cosh(a\psi)\right)(2n-1)g^{2n-2}}{g^{4n-2}}\\
    & = e^{-n\psi}\psi'g^{-2n}\left(\left(-nP_n + aY\frac{\partial
      P_n}{\partial x} + aX\frac{\partial P_n}{\partial
      y}\right)(aX+Y) - P_n(2n-1)(a^2Y+aX)\right)\\
    & = e^{-(n+1)\psi}g^{-2n-1}
    \left(P_n(-naX-nY-(2n-1)a^2Y-(2n-1)aX)\right)\\
    & + e^{-(n+1)\psi}g^{-2n-1}\left(\frac{\partial P_n}{\partial
      x}(a^2XY+aY^2) + \frac{\partial P_n}{\partial y}(a^2X^2 +
    aXY)\right)
  \end{align*}
  and we can conclude the desired result by the induction axiom.
\end{proof}

\begin{example}
  As an application of Proposition~\ref{derivative} we obtain the
  Taylor series formula of $\psi_0$ at $x=0$. We then get
  \[
  \psi_0^{(n)}(0)=\frac {P_n(1,0)}{a^{2n-1}}.
  \]
  In particular, the first few Taylor coefficients are
  \[
  \psi_0(0)=0,\quad \psi_0'(0)=\frac{1}{a},\quad \psi_0''(0)=-\frac{2}{a^2},
  \text{ and } \psi_0^{(3)}(0) = \frac {-a^4+10a^2}{a^5}.
  \]
  In Section~\ref{section7}, we shall give a more efficient method for
  computing these coefficients.
\end{example}

\begin{proposition}
  The function
  \begin{equation}\label{ppsi}
  \Psi(x) = x\psi(x)-\frac {x}{1-a^2}+\frac{ax}{1-a^2}\coth(a\psi(x))
  \end{equation}
  is a primitive function of $\psi(x)=\psi(a,x)$.
\end{proposition}

\begin{proof} Using the substitution $x=\sinh(at)e^t$,
  $\dd x=(a\cosh(at)+\sinh(at))e^t\dd t$, we get
  \begin{align*}
    \Psi(x) &= \int \psi(x)\dd x=\int t(a\cosh(at)+\sinh(at))e^t\dd
    t\\ & = \frac{1}{2}\int
    t\left((1+a)e^{(1+a)t}-(1-a)e^{(1-a)t}\right)\dd
    t\\ & = \frac{1}{2}\left(t-\frac{1}{1+a}\right)e^{t(1+a)} -
    \frac{1}{2}\left(t-\frac{1}{1-a}\right)e^{t(1-a)}\\ &=t\sinh
    (at)e^t+\frac{1}{1-a^2}(-e^t\sinh(at)+ae^t\cosh(at))\\ &=x\psi(x)
    -\frac{x}{1-a^2}+\frac{a}{1-a^2}e^{\psi(x)}\cosh(a\psi(x))\\ &=x\psi(x)
    -\frac{x}{1-a^2}+\frac{ax}{1-a^2}\coth(a\psi(x)).
  \end{align*}
\end{proof}

\begin{remark}
  In a similar manner, using the substitution $x=\sinh(at)e^t$, one
  can prove that integrals of the form
  \[
  \int x^n\psi^m(x)\dd x, \quad n,m\in \mathbb N,
  \]
  can be integrated in an elementary way.
\end{remark}

Exploiting the Taylor series $\psi_0(x) = \tfrac{x}{a} -
\tfrac{2x^2}{a^2}+\ldots$, and the value $\psi(L_a)=M_a$, we have
\begin{align*}
  \Psi_0(0) & = \lim_{x\to 0}\Psi_0(x) = \frac{a}{1-a^2},\\
  \Psi(L_a) & = L_a\left(M_a - \frac{2}{1-a^2}\right).
\end{align*}
Hence, we can now compute the following definite integral:
\begin{equation*}
  \int_{L_a}^0\psi_0(a,x)\dd x = \Psi_0(0) - \Psi(L_a) =
  \frac{a+2L_a}{1-a^2} - L_a M_a,
\end{equation*}
and, since $\psi_{-1}$ is the inverse of $f$,
\begin{align*}
  \int_{L_a}^0\psi_{-1}(a,x)\dd x
  & = \int_{-\infty}^{M_a}f(a,w)\dd w - L_aM_a \\
  & = \frac{2L_a}{1-a^2}-L_aM_a \\
  & = \frac{\left(\frac{1-a}{1+a}\right)^{\frac{1}{2a}}}{2\sqrt{1-a^2}}
  \left(\frac{-4a}{1-a^2} + \log\left(\frac{1-a}{1+a}\right)\right).
\end{align*}

%==============================================================================
\section{Explicit formulas for $\psi$}\label{section5}

In this section we shall, for certain rational values of $a$, find
explicit formulas for both branches of the $\psi$-function and the
transition function $\omega$.  We let $\mathrm{root}(p(x))$ denote the
real-valued positive solution(s) to the polynomial equation
$p(x)=0$. Recall that
\begin{equation}\label{def1}
  f(a,w) = e^w\sinh(aw)= \frac{1}{2}\left(e^{w(1+a)} - e^{w(1-a)}\right).
\end{equation}
For natural numbers $n$ and $m$ such that $n>m$ we set
\[
a = \frac{n-m}{n+m} \quad \text{and}\quad \alpha=\frac{2}{n+m}.
\]
so that
\[
1 + a = \frac{2n}{n+m}\text{ and } 1 - a = \frac{2m}{n+m}.
\]
Using the substitution $Y=e^{\alpha w}=e^{\frac{2w}{n+m}}$
we rewrite~\eqref{def1} as
\[f(a,w) = \frac{1}{2}\left( Y^n - Y^m\right) = x\]
and then
\begin{equation}\label{equ1}
  w = \frac{n+m}{2}\log\left(\text{root}(Y^n-Y^m-2x)\right).
\end{equation}
We are then able to give explicit solutions to this equation from
classical methods in the following cases:
\begin{equation*}
  w = \begin{cases}
    %\frac{3}{2}\log (\text{root}(Y^2-Y-2x)) &\text{if } a = 1/3,\\
    2\log (\text{root}(Y^3-Y-2x)) &\text{if } a = 1/2,\\
    \frac{5}{2}\log (\text{root}(Y^3-Y^2-2x))&\text{if } a = 1/5,\\
    \frac{5}{2}\log (\text{root}(Y^4-Y-2x)) &\text{if } a = 3/5,\\
    \frac{7}{2}\log (\text{root}(Y^4-Y^3-2x)) &\text{if } a = 1/7.
  \end{cases}
\end{equation*}
And the case $a=1/3$ has already been treated in
Section~\ref{section3}.

\medskip

\noindent\emph{Case a=1/2:} After some straightforward but rather
lengthy calculations using Cardano's method for solving cubic
equations, we arrive at the following:
\[
\psi_0(\tfrac{1}{2},x) =
\begin{cases}
  2\log \left(\sqrt[3]{x+\sqrt{x^2-\frac{1}{27}}} +
  \sqrt[3]{x-\sqrt{x^2-\frac{1}{27}}}\right), &
  \text{if } \frac{\sqrt{3}}{9}\le x\\
  2\log\left(\frac{2}{\sqrt{3}}
  \cos\left(\frac{\arccos(3\sqrt{3}x)}{3}\right) \right), &
  \text{if } -\frac{\sqrt{3}}{9} < x < \frac{\sqrt{3}}{9}
\end{cases}
\]
and
\[
\psi_{-1}(\tfrac{1}{2},x) = 2\log\left(\frac{2}{\sqrt{3}} \cos\left(
\frac{\arccos(3\sqrt{3}x)+4\pi}{3}\right) \right),
\quad\text{if }-\frac{\sqrt{3}}{9}<x< 0,
\]
which then gives the transition function
\[
\omega(\tfrac{1}{2},z) =
\begin{cases}
  2 \log \left(\frac{2}{\sqrt{3}} \cos \left(\frac{1}{3}
  \arccos\left(\frac{3\sqrt{3}}{2}
  \left(e^{3z/2} - e^{z/2}\right)\right)\right)\right), &
  \text{if } z < -\log 3\\
  2 \log \left(-\frac{2}{\sqrt{3}} \sin \left(\frac{\pi }{6}
  -\frac{1}{3}\arccos\left(\frac{3\sqrt{3}}{2}
  \left(e^{3x/2}-e^{x/2}\right)\right)\right)\right), &
  \text{if } -\log 3\le z < 0.
\end{cases}
\]

\medskip

\noindent\emph{Case a=1/5:} Again using Cardano's method for solving
cubic equations, we get in this case:
\[
\psi_0(x) =
\begin{cases}
  \frac{5}{2}\log\left(\sqrt[3]{x+\frac{1}{27} +
    \sqrt{x^2+\frac{2x}{27}}} + \sqrt[3]{x+\frac{1}{27} -
    \sqrt{x^2+\frac{2x}{27}}} + \frac{1}{3}\right), &
  \text{if } 0 \le x \\
  \frac{5}{2} \log\left(\frac{2}{3} \cos\left(
  \frac{\arccos(27x+1)}{3} \right) + \frac{1}{3} \right), &
  \text{if } -\frac{2}{27} < x < 0,
\end{cases}
\]
and
\[
\psi_{-1}(x) = \frac{5}{2} \log\left(\frac{2}{3} \cos\left(
\frac{\arccos(27x+1) + 4\pi}{3} \right) + \frac{1}{3} \right), \quad
\text{if } -\frac{2}{27} < x < 0.
\]
And now we obtain the transition function
\[
\omega(\tfrac{1}{5},z) =
\begin{cases}
  \frac{5}{2} \log \left(\frac{2}{3} \cos \left(\frac{1}{3} \arccos\left(\frac{27}{2}
  \left(e^{6 x/5}-e^{4 x/5}\right)+1\right)\right)+\frac{1}{3}\right), &
  \text{if }z<-\frac{5}{2}\log(3/2)\\
  \frac{5}{2} \log \left(\frac{1}{3}-\frac{2}{3} \sin \left(\frac{\pi}{6} - \frac{1}{3}
  \arccos\left(\frac{27}{2} \left(e^{6 x/5} - e^{4 x/5}\right)+1\right)\right)\right), &
  \text{if } -\frac{5}{2}\log(3/2)\le z < 0.
\end{cases}
\]

\medskip

In the following cases we will only state the $\psi$-functions.

\bigskip

\noindent\emph{Case a=3/5:} In this case we shall solve
\[
Y^4-Y-2x=0 \quad \text{for } x\geq -\frac{3}{8\sqrt[3]{4}}
\]
by using Ferrari's method. First we need the solution to the auxiliary
equation:
\[
v^3+2xv-\frac{1}{8} = 0.
\]
In our case, $x\geq -\frac{3}{8\sqrt[3]{4}}$, so that
\[
v = \frac{1}{\sqrt[3]{16}}\left(\sqrt[3]{1 - \sqrt{4\left(\frac{8x}{3}
    \right)^3 + 1}} + \sqrt[3]{1 + \sqrt{4\left(\frac{8x}{3} \right)^3
    + 1}} \right),
\]
and we then arrive at
\[
\psi_0(\tfrac{3}{5},x) = \frac{5}{2}\log\left( \frac{(2v)^{3/4} +
  \sqrt{2-(2v)^{3/2}}}{2\sqrt[4]{2v}} \right),
\quad\text{if } -\frac{3}{8\sqrt[3]{4}}\le x
\]
Furthermore,
\[
\psi_{-1}(\tfrac{3}{5},x) = \frac{5}{2} \log\left(\frac{(2v)^{3/4} -
  \sqrt{2-(2v)^{3/2}}}{2\sqrt[4]{2v}} \right),
\quad \text{if }-\frac{3}{8\sqrt[3]{4}} < x < 0.
\]

\medskip

\noindent\emph{Case a=1/7:} In this case we shall solve
\[
Y^4-Y^3-2x=0 \quad \text{for } x\geq -\left(\frac{3}{8}\right)^3
\]
by using Ferrari's method. After some lengthy but straightforward
calculations we get:
\[
\psi_0(\tfrac{1}{7},x) = \frac{7}{2}\log \left(\frac{\sqrt{2v+\frac{1}{4}} +
  \sqrt{-2v+\frac12 + \frac{1}{2\sqrt{8v+1}}}}{2} +\frac 14\right),
\quad\text{if } -\left(\frac{3}{8}\right)^3\le x
\]
and
\[
\psi_{-1}(\tfrac{1}{7},x) = \frac{7}{2} \log\left(\frac{\sqrt{2v+\frac{1}{4}} -
  \sqrt{-2v+\frac{1}{2} + \frac{1}{2\sqrt{8v+1}}}}{2} +\frac 14\right),
\quad\text{if } -\left(\frac{3}{8}\right)^3\le x < 0,
\]
where
\[
v = \sqrt[3]{-\frac{x}{8} + \sqrt{\left(\frac{2x}{3}\right)^3 +
    \left(\frac{x}{8}\right)^2 }} + \sqrt[3]{-\frac{x}{8} -
  \sqrt{\left(\frac{2x}{3}\right)^3 + \left(\frac{x}{8}\right)^2 }}.
\]

%==============================================================================
\section{A parametrization and a definite integral}\label{section6}

The main aim of this section is to prove that
\[
\int_{-\infty}^{0}\omega(a,z)\dd z = \frac{\pi^2}{3(a^2-1)}.
\]
We rely on an ingenious idea to simultaneously parameterize both
branches of $\psi$ for $x<0$. This idea originates
from~\cite{Corless:1996}, and is explained in~\cite[Theorem
  1.3.1]{Kalugin:2011}. Our parametrization is presented in the
following theorem:

\begin{theorem}\label{sec6_thm1}
  Let
  \[
  x =\frac{1}{2}\left(\frac{\alpha^{1+a} - \alpha^{2a}}{\alpha^{1+a}
    - 1}\right)^{\frac{1-a}{2a}} \left( \frac{1 -
    \alpha^{2a}}{\alpha^{1 + a} - 1} \right)
  \]
  where $1<\alpha$. Then
  \begin{align*}
    \psi_0(a,x)& =
    \log(\alpha) + \psi_{-1}(a,x) \\
    \psi_{-1}(a,x) & =
    \frac{1}{2a}\log\left(\frac{\alpha^{1-a} - 1}{\alpha^{1+a} -
      1}\right).
  \end{align*}
\end{theorem}

\begin{proof}
  We are looking for $\alpha$ such that
  \[
  x(\alpha) = e^{\psi_0(x(\alpha))}\sinh(a\psi_0(x(\alpha))) =
  e^{\psi_1(x(\alpha))}\sinh(a\psi_1(x(\alpha)))\, .
  \]
  Let $\psi_0= \psi_{-1} + p/a$, where $p>0$ is a parameter. Then,
  \[
  e^{\psi_0}\sinh(a\psi_0) = e^{\left(\psi_{-1} + \frac{p}{a}\right)}
  \sinh(a(\psi_{-1}+p/a)) = e^{\psi_{-1}} \sinh(a\psi_{-1})\, ,
  \]
  which implies
  \[
  e^{-\frac{p}{a}} = \frac{\sinh(a\psi_{-1}+p)}{\sinh(a\psi_{-1})} =
  \cosh(p) + \coth(a\psi_{-1})\sinh(p)\, .
  \]
  Hence,
  \[
  \psi_{-1} = \frac{1}{a}\operatorname{arcoth}\left(\frac{e^{-\frac{p}{a}}
    - \cosh(p)}{\sinh(p)} \right) = \frac{1}{2a}\log\left(\frac{e^{-p} -
    e^{-\frac{p}{a}}}{e^{p}-e^{-\frac{p}{a}}}\right)\, .
  \]
  Continuing in a similar manner, we set $\psi_{-1}=\psi_{0}-p/a$, where
  $p>0$ is a parameter, which implies
  \[
  e^{\psi_0-p/a}\sinh(a\psi_0-p)) = e^{\psi_0}\sinh(a\psi_0)\, .
  \]
  Hence,
  \[
  e^{\frac{p}{a}} = \frac{\sinh(a\psi_{0}-p)}{\sinh(a\psi_{0})} =
  \cosh(p)-\coth(a\psi_{0})\sinh(p)\, ,
  \]
  which yields
  \[
  \psi_{0} = \frac{1}{a}\operatorname{arcoth}\left(\frac{-e^{\frac{p}{a}}
    + \cosh(p)}{\sinh(p)}\right) = \frac{1}{2a}\log\left(\frac{-e^{p} +
    e^{\frac{p}{a}}}{-e^{-p} + e^{\frac{p}{a}}}\right)\, .
  \]
  Thus, for $p\in(0,\infty)$ it holds
  \[
  x = e^{\psi_0}\sinh(a\psi_0) = \frac{1}{2}\left(\frac{-e^{p} +
    e^{\frac{p}{a}}}{-e^{-p} + e^{\frac{p}{a}}} \right)^{\frac{1-a}{2a}}
  \left( \frac{e^{-p} - e^p}{e^{\frac{p}{a}} - e^{-p}}\right)\, .
  \]
  Set  $p=a\log(\alpha)$ where $\alpha\in (1,\infty)$. Then
  \begin{align*}
    x & =\frac{1}{2}\left(\frac{\alpha^{1+a} - \alpha^{2a}}{\alpha^{1+a}
      - 1}\right)^{\frac{1-a}{2a}} \left( \frac{1 -
      \alpha^{2a}}{\alpha^{1 + a} - 1} \right)\\ \psi_0(x)& =
    \log(\alpha) + \psi_{-1}(x) \\ \psi_{-1}(x) & =
    \frac{1}{2a}\log\left(\frac{\alpha^{1-a} - 1}{\alpha^{1+a} -
      1}\right)\, .
  \end{align*}
\end{proof}

%% \begin{remark}
%%   It follows from [15] that branches of the Lambert $W$ function can be parameterized in the following way
%%   \[
%%   W_0(x)=-\frac {p}{\sinh p}e^{-p},\ \ W_{-1}(x)=-\frac {p}{\sinh p}e^{p},
%%   \]
%% where
%% \[
%% x=-\frac {p}{\sinh p}e^{-p\coth p}, \ p\geq 0.
%% \]
%% After changing variables $\alpha=e^{-2p}\in(0,1]$ we get
%% \[
%% W_0=\frac{\log(\alpha)}{1-\alpha}, \ \ W_{-1}=\frac{\alpha\log(\alpha)}{1 - \alpha}.
%% \]

%% Now observe that if we let $a\to 0^+$, then
%%   \[
%%   \psi_0(x) \to \frac{\log(\alpha)}{1-\alpha} = W_0 \quad \text{and }
%%   \quad \psi_{-1}(x) \to \frac{\alpha\log(\alpha)}{1 - \alpha} = W_{-1},
%%   \]
%%   where $W_0$ and $W_{-1}$ denote the parametrization of the classical
%%   branches of the Lambert $W$ function.
%% \end{remark}

Thanks to the parametrization in Theorem~\ref{sec6_thm1} we get an
alternative way to find explicit formulas for $\psi_{0}$, and
$\psi_{-1}$. The below example is for the case $a=1/3$.

\begin{example}
  We return to the special case of $a=1/3$. For $-1/8\le x<0$ and
  $\alpha\geq 1$ we let
  \begin{equation}\label{eqAlpha}
    x = \frac{-\alpha^{2/3}}{2(\alpha^{2/3}+1)^2}\, ,
  \end{equation}
  and
  \[
  \psi_{0} = -\frac{3}{2}\log\left(1 + \alpha^{-2/3}\right) \quad
  \text{and } \psi_{-1} = -\frac{3}{2}\log\left(1+\alpha^{2/3}\right)\,.
  \]
  Thus, by solving~(\ref{eqAlpha}) in $\alpha^{2/3}$, and using
  Theorem~\ref{sec6_thm1} we get
  \[
  \psi_{0} = \frac{3}{2}\log\left(\frac{1}{2} + \sqrt{2x+\frac14}\right)
  \quad \text{and } \quad \psi_{-1} = \frac{3}{2}\log \left(\frac{1}{2}
  - \sqrt{2x+\frac14}\right)\,.
  \]
  \hfill{$\Box$}
\end{example}

Employing Theorem~\ref{sec6_thm1}, we can now state the following theorem:

\begin{theorem}
  \[
  \int_{-\infty}^{0}\omega(a,z)\dd z = \frac{\pi^2}{3(a^2-1)}.
  \]
\end{theorem}
\begin{proof}
  The definition of the transition function states:
  \[
  \omega(a,z)=
  \begin{cases}
    \psi_{0}(a,f(a,z)) & \text{if } z <M_a \\
    \psi_{-1}(a,f(a,z)) & \text{if } M_a\leq z < 0\, .
  \end{cases}
  \]
  For $y\geq L_a$ in the relation $y=f(a,z)$ we note that one inverse
  branch is $z=\psi_0(a,y)\ge M_a$, and the other branch is
  $z=\psi_{-1}(a,y)\le M_a$.  Hence,
  \[
  \int_{-\infty}^{0}\omega(a,z)\dd z = \int_{0}^{L_a} \left( \psi_0(y)
  \psi'_{-1}(y) - \psi_{-1}(y) \psi'_0(y) \right) \dd y.
  \]
  From Theorem~\ref{sec6_thm1} it follows
  \[
  \psi_0(y(\alpha)) = \log(\alpha) + \psi_{-1}(y(\alpha)) \quad
  \text{and} \quad \psi_{-1}(y(\alpha)) =
  \frac{1}{2a}\log\left(\frac{\alpha^{1-a} - 1}{\alpha^{1+a} - 1}\right),
  \]
  and therefore
  \[
  \int_{-\infty}^{0}\omega(a,z)\dd z = \int_{1}^{\infty}
  \left(\frac{1}{\alpha}\Psi(\alpha) -
  \log(\alpha)\Psi'(\alpha)\right)\, \dd \alpha ,
  \]
  where
  \[
  \Psi(\alpha)=\frac{1}{2a}\log\left(\frac{\alpha^{1-a} -
    1}{\alpha^{1+a} - 1}\right)\, .
  \]
  Finally, from the substitution
  $\alpha=\frac{1}{\beta}$ it follows
  \begin{multline*}
    \int_{-\infty}^{0}\omega(a,z)\dd z =
    \int_{0}^{1}\Bigg(\frac{1}{2a\beta}\sum_{n=1}^{\infty}
    \frac{-\left((\beta^{1-a})^n - (\beta^{1+a})^n\right)}{n} \\ +
    \frac{\log(\beta)}{2a\beta}\left(
    (1-a)\sum_{n=1}^{\infty}(\beta^{1-a})^n -
    (1+a)\sum_{n=1}^{\infty}(\beta^{1+a})^n\right)\Bigg)\, \dd \beta
    \\ = \frac{1}{1-a^2}\left( - \frac{\pi^2}{6} -
    \frac{\pi^2}{6}\right) = \frac{\pi^2}{3(a^2-1)}\, .
  \end{multline*}
\end{proof}

%==============================================================================
\section{Taylor series of $\psi_0$ at zero}\label{section7}

In this section and the next, we will focus on series expansions.
Here we use Lagrange's inversion theorem to determine the Taylor
series for $\psi_0$ about $x=0$ (see~(\ref{Sec7_TaylorPsio})). First
let us recall Lagrange's inversion theorem~\cite{AbramowitzStegun:1972,Charalambides:2002}:

\begin{theorem}[Lagrange's inversion theorem]\label{thm_LIT}
  If $x=f(w)$, $x_0=f(w_0)$, for some real-analytic function in a
  neighbourhood of $w_0$ with $f^{\prime}(w_0)\neq 0$, then
  \[
  w=g(x)=w_0 +
  \sum_{k=1}^{\infty}\frac{x-x_0}{k!}\left. \left(\frac{d^{k-1}}{dw^{k-1}}
  \left(\frac{w-w_0}{f(w)-x_0} \right)^k \right) \right|_{w=w_0},
  \]
  where the convergence radius is strictly positive. Furthermore, if
  \[
  f(w)=\sum_{k=0}^{\infty} f_k\frac{w^k}{k!} \quad \text{ and } \quad
  g(x)=\sum_{k=0}^{\infty} g_k\frac{x^k}{k!}\, ,
  \]
  then it holds
  \begin{align*}
    &g_n = \frac{1}{f_1^n}\sum_{k=1}^{n-1}(-1)^kn^{\bar k} B_{n-1,k}
    \left(\frac{f_{2}}{2f_1}, \frac{f_3}{3f_1},\dots,
    \frac{f_{n-k+1}}{(n-k+1)f_1}\right)\\ & = \frac{1}{a^n}
    \sum_{k=1}^{n-1}(-1)^ka^kn^{\bar k} B_{n-1,k}\left( \frac{f_{2}}{2},
    \frac{f_3}{3},\dots, \frac{f_{n-k+1}}{(n-k+1)}\right), \,n\geq 2,
  \end{align*}
  where $g_1=\frac {1}{f_1}=\frac 1a$, $n^{\bar k}=n(n+1)\cdots(n+k-1)$ is
  the rising factorial and $B_{n,k}$ is the partial exponential Bell polynomial.
\end{theorem}

We apply this theorem to the function
\begin{align*}
  f(w) &= f(a,w) = \frac{1}{2} (e^{(1+a)w} - e^{(1-a)w})\\
  & = \sum_{n=1}^{\infty}\frac{(1+a)^n-(1-a)^n}{2n!} w^n \\
  & = \sum_{n=1}^{\infty}\frac{f_n}{n!}w^n,
\end{align*}
and then we arrive at
\[
\psi_0(a,x)=\sum_{n=1}^{\infty}\frac{g_n}{n!}x^n,
\]
where
\[
g_n = \frac{1}{f_1^n}\sum_{k=1}^{n-1}(-1)^kn^{\bar k} B_{n-1,k}
\left(\frac{f_{2}}{2f_1}, \frac{f_3}{3f_1},\dots,
\frac{f_{n-k+1}}{(n-k+1)f_1}\right)\, .
\]

Thus, the first few terms of the Taylor series of $\psi_0$ around $x=0$ are given by:
\begin{multline}\label{Sec7_TaylorPsio}
  \psi_0(a, x) = \frac{x}{a} -\frac{1}{a^2}x^2 -\frac{(a^2-9)}{6a^3}x^3
  +\frac{2(a^2-4)}{3a^4}x^4 \\
  +\frac{(9a^4-250a^2+625)}{120a^5}x^5
  -\frac{2(4a^4-45a^2+81)}{15a^6}x^6\\
  -\frac{(225a^6-12691a^4+84035a^2-117649)}{5040a^7}x^7
  +\frac{16(9a^6-196a^4+896a^2-1024)}{315a^8}x^8 \\
  +\frac{(1225a^8-116244a^6+1439046a^4-4960116a^2+4782969)}{40320a^9}x^9 \\
  -\frac{2(576a^8-20500a^6+170625a^4-468750a^2+390625)}{2835a^{10}}x^{10}
  +\mathcal{O}(x^{11})
\end{multline}
Furthermore, the radius of converges of the Taylor series can not exceed $|L_a|$.

%==============================================================================
\section{Series expansions of $\psi$ at $x=L_a$ and of $\omega$ at $z=M_a$}\label{section8}

This section aims to determine the series expansions of $\psi$ at
$x=L_a$ and of $\omega$ at $z=M_a$ with the help of
Proposition~\ref{invthm} as a tool. The series are stated
in~(\ref{seriespsi0}),~(\ref{seriespsi1}), and~(\ref{seriesomega}),
respectively.

\begin{proposition}\label{invthm}
  Let $f$ be a smooth function in the neighborhood of $0$ such that
  $f(0)=f'(0)=0$ and $f''(0)>0$. Then $f$ has two smooth inverse
  functions (in some neighborhood of $0$), a right inverse $h(x)>0$
  for $x>0$, and a left inverse $g(x)<0$ for $x>0$, which can be
  expressed as
  \[
  \begin{aligned}
    &h(x)=h_1\sqrt{x}+h_2x+h_3x^{3/2}+\ldots, \\
    &g(x)=g_1\sqrt{x}+g_2x+g_3x^{3/2}+\ldots,
  \end{aligned}
  \]
  where $g_{2k}=h_{2k}$ and $g_{2k-1}=-h_{2k-1}$ for $k=1,2,\ldots$.
\end{proposition}

\begin{proof}
  Without loss of generality, we may assume that $f$ is convex in a
  possible smaller neighborhood of zero.  The Taylor expansion at zero
  can be presented as
  \[
  f(x)=x^2(f_2+f_3x+\cdots), \quad \text{with } f_2>0.
  \]
  Therefore, $f(x)=x^2\tilde f(x)$ for some smooth
  function $\tilde f$ with $\tilde f(x)>0$. Thus, $\sqrt {\tilde
    f(x)}$ is smooth in the neighborhood of $0$. Set
  \[
  F(x)=\begin{cases}
  \sqrt{x^2\tilde f(x)}, &\text{if } x\geq 0 \\
  -\sqrt{x^2\tilde f(x)}, &\text{if } x\leq 0 .
  \end{cases}
  \]
  The function $F$ is invertible and smooth in a possible punctured
  neighborhood of $0$.  Furthermore, for all $n\in \mathbb N$ the
  following limits exist:
  \[
  \lim_{x\to 0^+}f^{(n)}(x) \quad \text{ and } \quad \lim_{x\to 0^-} f^{(n)}(x)
  \]
  and $F^2=f$. From this it follows that $f$ has two smooth inverse
  functions, $f^{-1}=F^{-1}\circ (x^2)^{-1}$, which can be expressed
  as $F^{-1}(\sqrt{y})$ and $F^{-1}(-\sqrt{y})$, $y\geq 0$. Finally,
  by applying the Taylor series expansion to $F^{-1}$ we get the
  desired conclusion.
\end{proof}

Let $K_a = L_a(a^2-1)$ and recall that $f(a,M_a) = L_a$. Then we
receive the following series expansion around $y=0$
\begin{multline*}
  \frac{1}{K_a}\left(f(a,y+M_a)-L_a\right) = \frac{1}{2}y^2 +
  \frac{1}{3}y^3 + \frac{a^2+3}{24}y^4 + \frac{a^2+1}{30}y^5 \\ +
  \frac{a^4+10a^2+5}{720}y^6 + \frac{3a^4+10a^2+3}{2520}y^7 +
  \frac{a^6 + 21a^4 +35a^2+7}{40320}y^8 + \mathcal{O}(y^9)
\end{multline*}
The inverse series expansion gives us $\psi_0$ can be found by using
the series ansatz of $f$ and $h$ in Proposition~\ref{invthm} and solve
$f(h(x))=x$ term by term.
\begin{multline}\label{seriespsi0}
  \psi_0(a, x\,K_a + L_a) - M_a = \sqrt{2}\sqrt{x} - \frac{2}{3}x +
  \frac{11-3a^2}{18\sqrt{2}}x^{3/2} - \frac{43-27a^2}{135}x^2 \\ +
  \frac{81a^4 - 786a^2 + 769}{2160\sqrt{2}} x^{5/2} -
  \frac{8(81a^4-318a^2+221)}{8505} x^3 \\ +
  \frac{680863-1273509a^2+551853a^4-30375a^6}{2721600\sqrt{2}}x^{7/2}
  + \mathcal{O}(x^4)
\end{multline}
which converges for $0\le x< -L_a/K_a=1/(1-a^2)$.  Note that the
coefficients of $x^3$ and $x^{7/2}$ can change sign for $0<a<1$.

Thanks to Proposition~\ref{invthm} we can deduce the other branch
\begin{multline}\label{seriespsi1}
  \psi_{-1}(a, x\,K_a + L_a) - M_a = -\sqrt{2}\sqrt{x} - \frac{2}{3}x
  - \frac{11-3a^2}{18\sqrt{2}}x^{3/2} - \frac{43-27a^2}{135}x^2 \\ -
  \frac{81a^4 - 786a^2 + 769}{2160\sqrt{2}}x^{5/2} -
  \frac{8(81a^4-318a^2+221)}{8505} x^3 \\ -
  \frac{680863-1273509a^2+551853a^4-30375a^6}{2721600\sqrt{2}}x^{7/2}
  + \mathcal{O}(x^4)
\end{multline}
and again this series converges for $0\le x< 1/(1-a^2)$. Note that the
second argument in both series is $x\,K_a + L_a\in [L_a,0)$.

Now we find the composition of the series of $\psi_{-1}(a, x\,K_a +
L_a) - M_a$ with that of $(f(a,x+M_a)-L_a)/K_a$, which gives us the
series of $\omega(a,x+M_a)-M_a$, and note that the $K_a$-factor
cancels out,

\begin{multline}\label{seriesomega}
  \omega(a,x+M_a)-M_a = -x - \frac{2}{3}x^2 - \frac{4}{9}x^3
  +\frac{2(3a^2-22)}{135}x^4 +\frac{4(9a^2-26)}{405}x^5 \\ -
  \frac{4(3a^4-88a^2+150)}{2835}x^6 -
  \frac{8(81a^4-808a^2+956)}{42525}x^7\\ +
  \frac{2(27a^6-2106a^4+11124a^2-9968)}{127575}x^8 +
  \frac{4(135a^6-3258a^4+11064a^2-7928)}{229635}x^9\\ -
  \frac{4(2025a^8-341604
    a^6+4049838a^4-9850752a^2+5857336)}{189448875}x^{10}
  +\mathcal{O}(x^{11}).
\end{multline}

Note that the coefficient of $x^4$ is the first which depends on $a$.
Note also that the coefficient of $x^9$ is the first which can become
zero for some $0<a<1$. This happens at $a=0.998689$. The coefficient
for $x^{10}$ becomes zero at $a=0.952489$.\\

\begin{example}
  In Fig.~\ref{fig5} we show $\psi_0$, $\psi_{-1}$, $\omega$ for
  $a=1/2$ and their respective series expansion as just stated. In
  this case the first few terms of the respective expansions are
  \begin{align*}
    \psi_0\left(\tfrac{1}{2},\tfrac{x}{4\sqrt{3}}-\tfrac{1}{3\sqrt{3}}\right)
    + \log(3) &= \sqrt{2}\sqrt{x} -
    \frac{2}{3}x+\frac{41}{72\sqrt{2}}x^{3/2} - \frac{29}{108}x^2 +
    \frac{9241}{34560\sqrt{2}}x^{5/2}
    +\ldots\\ \psi_{-1}\left(\tfrac{1}{2},\tfrac{x}{4\sqrt{3}}-\tfrac{1}{3\sqrt{3}}\right)
    + \log(3) &= -\sqrt{2}\sqrt{x} -
    \frac{2}{3}x-\frac{41}{72\sqrt{2}}x^{3/2} - \frac{29}{108}x^2 -
    \frac{9241}{34560\sqrt{2}}x^{5/2}
    +\ldots\\ \omega\left(\tfrac{1}{2},x-\log(3)\right)+\log(3) &= -x
    -\frac{2}{3}x^2 - \frac{4}{9}x^3 -\frac{17}{54}x^4
    -\frac{19}{81}x^5+\ldots
  \end{align*}\hfill{$\Box$}
\end{example}

\begin{figure}
  \begin{minipage}{0.33\textwidth}
    \begin{center}
      \includegraphics[width=\textwidth]{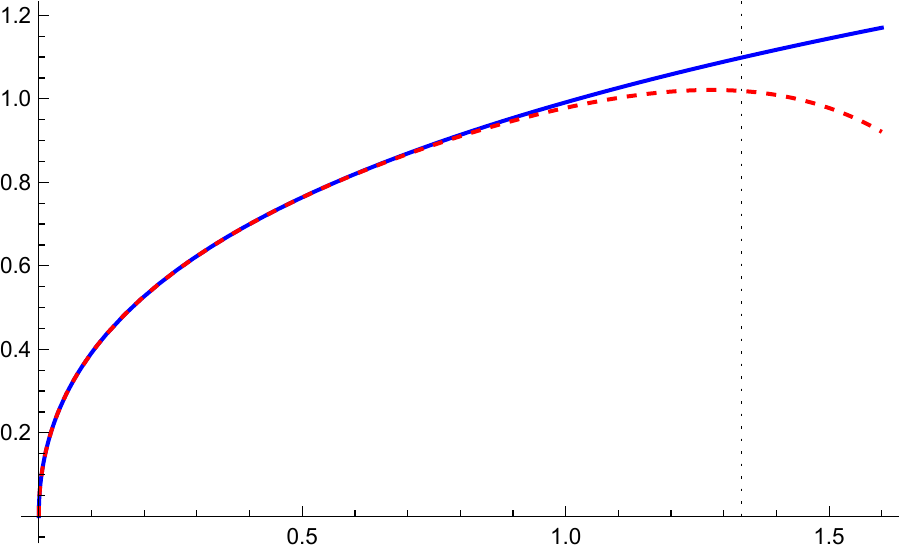}
    \end{center}
  \end{minipage}%
  \begin{minipage}{0.33\textwidth}
    \begin{center}
      \includegraphics[width=\textwidth]{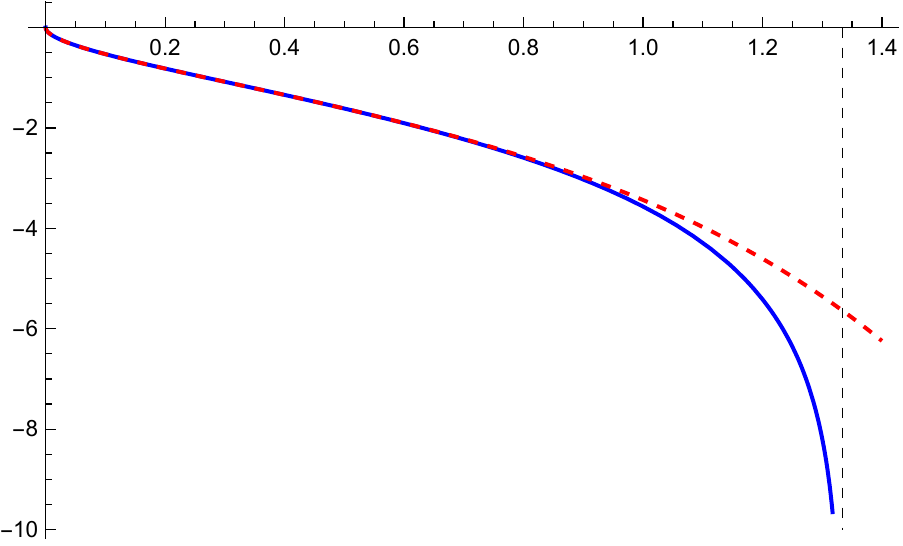}
    \end{center}
  \end{minipage}%
  \begin{minipage}{0.33\textwidth}
    \begin{center}
      \includegraphics[width=\textwidth]{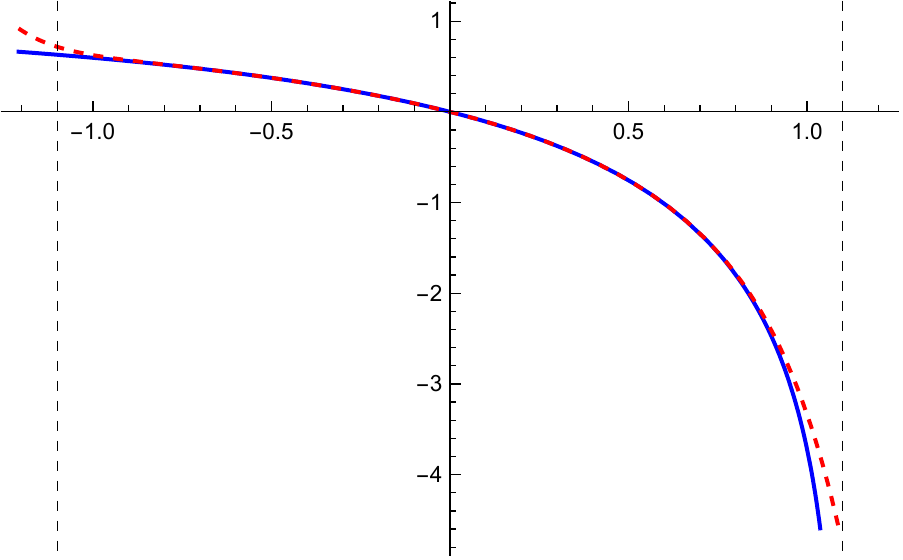}
    \end{center}
  \end{minipage}
  \caption{\label{fig5} Left: Plots of $\psi_0(a,x/K_a+L_a)+M_a$
    (solid blue curve) and the corresponding 6th order series
    expansion (dashed red curve).  Middle: Plots of
    $\psi_{-1}(a,x/K_a+L_a)+M_a$ (solid blue curve) and the
    corresponding 6th order series expansion (dashed red curve).
    Right: Plots of $\omega(a,x-M_a)+M_a$ (solid blue curve) and the
    corresponding 11th order series expansion. Plots show $a=1/2$. The
    dotted vertical lines in the plots indicate the end-points of the
    convergence interval. These are at $x=-L_a/K_a$ for the
    left\&middle plots and at $x=\pm M_a$ in the right plot.}
\end{figure}

%==============================================================================
\section{Asymptotics of the branch $\psi_0$ as $x\to\infty$}\label{section9}

Here we determine the asymptotic behaviour of $\psi_0(a,x)$ as
$x\to\infty$. An important tool in doing so is the following theorem.

\begin{theorem}\label{thmpsi0}
  Let
  \[x\ge \frac{1}{2}\left(1 +
  \frac{1}{e^{|\frac{1}{3}-a|}-1}\right)^{\frac{1+a}{2a}}.
  \]
  Then the following holds
  \begin{enumerate}
  \item if $0< a < \frac{1}{3}$, then
    \[
    \frac {1}{1+a}(2x)^{-\frac{2a}{1+a}} + \frac{\log(2x)}{1+a} \leq
    \psi_0(a,x)\leq \frac{\log(2x)}{1+a} +
    \frac{2}{1+a}(2x)^{-\frac{2a}{1+a}}.
    \]
     \item if $\frac 13<a<1$, then
    \[
    \frac{1}{3(1+a)}(2x)^{-\frac{2a}{1+a}} + \frac{\log(2x)}{1+a} \leq
    \psi_0(a,x) \leq \frac{\log(2x)}{1+a} +
    \frac{1}{1+a}(2x)^{-\frac{2a}{1+a}};
    \]
  \end{enumerate}
\end{theorem}

\begin{remark}
  Let $x_0$ be the point such that the bounds hold when $x\ge x_0$.
  Numerical experimentation suggests that for the lower bounds we have
  \begin{itemize}\itemsep2mm
  \item[$(i)$] If $0<a<1/3$
    then \[x_0=\frac{e^{-29/4}}{(\tfrac{1}{3}-a)^2}\] This seems to be
    the exact optimal point when $a\to(1/3)^-$. It is not optimal when
    $a\to 0^+$ but it still works then.
  \item[$(ii)$] If $1/3<a<1$ then $x_0\approx 0.1$ works. The difference
    between $\psi_0$ and the lower bound is decreasing when $x\ge 1/4$
    (this is the limit when $a\to 1^-$).
  \end{itemize}

  For the upper bounds we get
  \begin{itemize}\itemsep2mm
  \item[$(i)$] If $0<a\lesssim 0.102$ then $x_0\sim
    \exp(-3.8+\tfrac{0.117}{a})$.  The difference between the upper
    bound and $\psi_0$ is decreasing for
    $x>\exp(-3.1+\tfrac{0.35}{a})$. Both of these formulas are rough
    estimates but they are much smaller than the $x$-bound in the
    theorem. In any case, the upper bounds require very large $x$ if
    they are to hold for small $a$.
  \item[$(ii)$] When $0.102\lesssim a < 1$ the upper bound holds for all $x>0$.
  \item[$(iii)$] When $0.1255\lesssim a < 1$ the difference between the upper
    bound and $\psi_0$ is decreasing for all $x>0$.
  \end{itemize}
\end{remark}

\medskip

Let us start with observing that if $y=\psi_0(a,x)$, then
\[
x = \sinh(ay)e^y = \frac{e^{(1+a)y}-e^{(1-a)y}}{2} < \frac {e^{(1+a)y}}{2},
\]
and therefore we get
\begin{equation*}%\label{psi0_1}
  \psi_0(a,x) = y > \frac{1}{1+a} \log(2x).
\end{equation*}

We shall next continue with proving two auxiliary lemmas.
\begin{lemma}\label{4.1}
  For $0<a<1$ it holds
  \[
  0\leq \psi_0(a,x)-\frac{\log (2x)}{1+a}\leq
  \log\left(\frac{(2x)^{\frac{a}{1+a}}}{(2x)^{\frac{a}{1+a}} -
    (2x)^{-\frac{a}{1+a}}}\right),
  \]
  for all $x>0$. In particular,
  \begin{equation*}%\label{psi0_2}
    \lim_{x\to \infty}\left(\psi_0(a,x) - \frac{1}{1+a}\log(2x)\right) = 0.
  \end{equation*}
\end{lemma}

\begin{proof}
  Using the definition of $\psi_0$ and the fact that function $\sinh$
  is increasing we have
  \begin{align*}\label{psi0_3}
    &1\leq e^{\psi_0(a,x)-\frac{\log (2x)}{1+a}} =
    e^{\psi_0(a,x)}(2x)^{-\frac{1}{1+a}} =
    \frac{x(2x)^{-\frac{1}{1+a}}}{\sinh(a\psi_0(a,x))}\\ &\leq
    \frac{x(2x)^{-\frac{1}{1+a}}}{\sinh\left(a\frac{\log(2x)}{1+a}\right)}
    = \frac{(2x)^{\frac{a}{1+a}}}{(2x)^{\frac{a}{1+a}} -
      (2x)^{-\frac{a}{1+a}}}.
  \end{align*}
  The last term above tends to $1$ as $x\to \infty$, which proves the
  final claim.
\end{proof}

The following lemma is a straight-forward exercise obtained from
taking the Taylor expansion of $e^y$ about $0$.

\begin{lemma}\label{thmpsi0_lem}
  For any $0<a < \frac{1}{3}$ there exists a point $b(a)>0$ such that
  \begin{enumerate}
    \itemsep 2mm
  \item[$(i)$] $e^{2ay}-e^{(a-1)y}-(a+1)y\leq 0$, and
  \item[$(ii)$] $e^{2ay}-e^{(a-1)y}-(a+1)y+\frac{1}{2}(-3a^2-2a+1)y^2 \geq 0$,
  \end{enumerate}
  for all $0\leq y\leq b(a)$. Furthermore, for $a>\frac{1}{3}$ there
  exists a point $b(a)>0$ such that
  \[
  (1+a)y\leq e^{2ay} - e^{(a-1)y}\leq 3(1+a)y
  \]
  for all $0\leq y\leq b(a)$. In fact in both cases one can take $b(a)=|a-\frac 13|$.
\end{lemma}

We are now ready to present the proof of the main result of this section.
\begin{proof}[Proof of Theorem~\ref{thmpsi0}.]
  From Lemma~\ref{4.1} it follows that $\varphi(x) = \psi_0(a,x) -
  \frac{\log(2x)}{1+a}$ satisfies: $\varphi > 0$, $\varphi(x) \to 0$,
  as $x \to \infty$ and $\varphi(x) \leq h(x)$, where
  \[
  h(x)=\log\left(\frac{(2x)^{\frac{a}{1+a}}}{(2x)^{\frac{a}{1+a}} -
    (2x)^{-\frac{a}{1+a}}}\right)\, .
  \]
  Next, substitute $x=\varphi(x)$ in the implicit definition of $\psi_0$, and we arrive at
  \[
  e^{(1+a)\left(\varphi(x) + \frac{\log(2x)}{1+a}\right)} -
  e^{(1-a)\left(\varphi(x) + \frac{\log(2x)}{1+a}\right)} = 2x,
  \]
  and after a rearrangement we obtain
  \begin{equation}\label{thmpsi0_2}
    e^{2a\varphi(x)} - e^{(a-1)\varphi(x)} = (2x)^{-\frac{2a}{1+a}}.
  \end{equation}

  We have that $0<\varphi(x) \leq h(x)$. Furthermore, by inspection,
  for any $\epsilon>0$ we have that $h(x)\leq \epsilon$ if, and only
  if,
  \begin{equation}\label{estim}
  x\geq \frac 12\left(1+\frac{1}{e^{\epsilon}-1}\right)^{\frac {1+a}{2a}}.
  \end{equation}
  Choose $\epsilon=b(a)=|\frac{1}{3} - a|$, where $b(a)$ is from
  Lemma~\ref{thmpsi0_lem}. To conclude this proof we apply
  Lemma~\ref{thmpsi0_lem} with $y=\varphi(x)$ to~(\ref{thmpsi0_2}). To
  summarize and write out certain details:

  \begin{enumerate}
  \item  By Lemma~\ref{thmpsi0_lem}
    \begin{align*}
      &\varphi(x)\geq \frac{1}{1+a}(2x)^{-\frac{2a}{1+a}}, \text{ for
      } a < \frac{1}{3};\\ & \frac{1}{3(1+a)}(2x)^{-\frac{2a}{1+a}}
      \leq \varphi(x)\leq \frac{1}{1+a}(2x)^{-\frac{2a}{1+a}}, \text{
        for } a \geq \frac{1}{3}.
    \end{align*}
  \item For $a < \frac{1}{3}$, again by using Lemma~\ref{thmpsi0_lem}
    \[
    \beta_1(a)\varphi(x) + \frac{1}{2}\beta_2(a)\varphi^2(x) \leq
    (2x)^{-\frac{2a}{1+a}},
    \]
    where $\beta_1(a)=1+a$ and $\beta_2(a)=3a^2+2a-1$.
    Note that by (\ref{estim})
    \[
    \Delta = \beta_1(a)^2+2\beta_2(a)(2x)^{-\frac{2a}{1+a}} =
    (1+a)^2\left(1 + \frac{6a-2}{1+a}(2x)^{-\frac{2a}{1+a}}\right)
    \geq 0\, ,
    \]
    since
    \[
     1+\frac{1}{e^{|a-\frac 13|}-1} \geq \left|\frac{1+a}{2-6a}\right|.
    \]
    In the case $a < \frac{1}{3}$, we have that $\beta_2(a) =
    (3a-1)(a+1) < 0$ and
    \[
    \varphi_1 = \frac{\beta_1(a) + \sqrt{\Delta}}{-\beta_2(a)} >
    \varphi_2 = \frac{\beta_1(a) - \sqrt{\Delta}}{-\beta_2(a)} > 0.
    \]
    Therefore, $\varphi \leq \varphi_2$. Thus,
    \[
    \varphi(x) \leq \frac{1 - \sqrt{1 -
        \frac{2-6a}{1+a}(2x)^{-\frac{2a}{1+a}}}}{1 - 3a} \leq
    \frac{2}{1+a}(2x)^{-\frac{2a}{1+a}}.
    \]
  \end{enumerate}
\end{proof}

As a corollary of the previous proof, we can get a series expansion in
the variable $Y=(2x)^{-\frac{2a}{1+a}}$. First we
re-state~(\ref{thmpsi0_2}):
\[
e^{2a\varphi(x)}-e^{(a-1) \varphi(x)} = (2x)^{-\frac{2a}{1+a}},
\]
and then we use the Taylor expansion of the exponential function about $0$:
\[
\sum_{n=1}^{\infty}\frac{(2a)^k - (a-1)^k}{k!} \varphi^k(x) =
\sum_{n=1}^{\infty} \frac{\beta_k(a)}{k!} \varphi^k(x) =
(2x)^{-\frac{2a}{1+a}},
\]
where $\beta_k(a) = (2a)^k-(a-1)^k$, $\beta_1(a) = 1+a$, $\beta_2(a) =
3a^2+2a-1$, $\beta_3(a) = 7a^3+3a^2-3a+1$. Then by using
Theorem~\ref{thm_LIT} we arrive at
\[
\psi(a,x) = \sum_{n=1}^{\infty}\frac{g_n(a)}{n!}Y^n,
\]
where
\[
g_n(a) = \frac{1}{\beta_1^n} \sum_{k=1}^{n-1}(-1)^kn^{(k)}
B_{n-1,k}\left(\frac{\beta_{2}}{2\beta_1},
\frac{\beta_3}{3\beta_1},\dots,
\frac{\beta_{n-k+1}}{(n-k+1)\beta_1}\right),\, n\geq 2,
\]
and $g_1(a) = \frac{1}{\beta_1} = \frac{1}{a+1}$, $n^{\bar k} =
n(n+1)\cdots(n+k-1)$ is the rising factorial and $B_{n,k}$ denote the
partial or incomplete exponential Bell polynomials.

Let us compute a few coefficients:
\begin{align*}
  g_1(a) & = \frac{1}{\beta_1} = \frac{1}{a+1}\\
  g_2(a) & = \frac{1}{\beta_1^2(a)}(-1) 2 B_{1,1}(\hat{\beta_1}) =
  -\frac{\beta_2(a)}{\beta_1^3(a)} = \frac{1-3a}{(a+1)^2}\\
  g_3(a) & = \frac{1}{\beta_1^3(a)} \left(-3B_{2,1}(\hat{\beta_1},
  \hat{\beta_2}) + 12 B_{2,2}(\hat{\beta_1})\right) \\ & =
  \frac{1}{\beta_1^3(a)}\left(-3\hat{\beta_2} + 12
  \hat{\beta_1}^2\right) = \frac{-\beta_3\beta_1 +
    3\beta_2^2}{\beta_1^5} \\ & = \frac{2(10a^2-7a+1)}{(a+1)^3}.
\end{align*}
Thus,
\[
\varphi(x)\thicksim g_1(a)Y + \frac{g_2(a)}{2}Y^2 +
\frac{g_3(a)}{6}Y^3 = \frac{1}{a+1}Y + \frac {1-3a}{2(a+1)^2}Y^2 +
\frac{10a^2-7a+1}{3(a+1)^3}Y^3,
\]
which implies
\begin{equation*}%\label{asympsi0}
  \psi_0(a,x) \thicksim \frac{\log(2x)}{1+a} +
  \frac{1}{a+1}(2x)^{-\frac{2a}{1+a}} +
  \frac{1-3a}{2(a+1)^2}(2x)^{-\frac{4a}{1+a}} +
  \frac{10a^2-7a+1}{3(a+1)^3}(2x)^{-\frac{6a}{1+a}}.
\end{equation*}

%==============================================================================
\section{Asymptotics of the branch $\psi_{-1}$ as $x\to 0^-$}\label{section10}

Using the same techniques as in the previous sections, we shall study
the branch $\psi_{-1}(a,x)$'s behavior near zero. We shall prove the
following theorem. Recall that the branch $\psi_{-1}(a,x)$ is defined
for $L_a\le x<0$.

\medskip

\begin{theorem}\label{thmpsi1}
  Let $0<a<1$. Then
  \begin{enumerate}\itemsep2mm
  \item $\displaystyle{\lim_{x\to 0^-}\left(\psi_{-1}(a,x) -
    \frac{1}{1-a}\log(-2x)\right) = 0.}$
  \item $\displaystyle{\psi_{-1}(a,x) \geq \frac{\log(-2x)}{1-a} -
    \log\left(1 - (-2x)^{\frac{2a}{1-a}}\right).}$
  \item For $-\frac{1}{2}
    \left(\frac{1-a}{6a+2}\right)^{\frac{1-a}{2a}}\leq x\leq 0 $ it
    holds
    \[
    \frac{\log(-2x)}{1-a} + \frac{1}{1-a}(-2x)^{\frac{2a}{1-a}} \leq
    \psi_{-1}(a,x) \leq \frac{\log(-2x)}{1-a} +
    \frac{2}{1-a}(-2x)^{\frac{2a}{1-a}}.
    \]
  \end{enumerate}
\end{theorem}

In the proof of Theorem~\ref{thmpsi1} we shall need the following
elementary facts.

\begin{lemma}\label{thmpsi1_lem} For $x>0$ it holds
  \begin{enumerate}\itemsep2mm
  \item[$(i)$] $e^{-2ax} - e^{(-a-1)x} - (1-a)x \leq 0$ and
  \item[$(ii)$] $e^{-2ax} - e^{(-a-1)x} - (1-a)x - \frac{1}{2}(3a^2-2a-1)x^2 \geq 0$.
  \end{enumerate}
\end{lemma}

\begin{proof}[Proof of Theorem~\ref{thmpsi1}.]
  (1) If $y=\psi_{-1}(x)$, then
  \begin{equation}\label{psi-1_1}
    x = \sinh(ay)e^y = \frac{e^{(1+a)y} - e^{(1-a)y}}{2} > -\frac{e^{(1-a)y}}{2},
  \end{equation}
  and therefore we get
  \begin{equation}\label{psi-1_2}
    \psi_{-1}(a,x) = y > \frac{1}{1-a}\log(-2x).
  \end{equation}
  Then, let $\varphi(x) = \psi_{-1}(a,x) - \frac{\log(-2x)}{1-a}$. We
  know that $\varphi > 0$. From~(\ref{psi-1_1}) it follows
  \[
  e^{(1+a)\left(\varphi(x) + \frac{\log(-2x)}{1-a}\right)} -
  e^{(1-a)\left(\varphi(x) + \frac{\log(-2x)}{1-a}\right)} = 2x,
  \]
  and after rearranging we have
  \begin{equation}\label{est-13}
    e^{-2a\varphi(x)} - e^{(-a-1)\varphi(x)} = (-2x)^{\frac{2a}{1-a}}.
  \end{equation}
  Each term of the left side is bounded and since the right hand side
  tends to zero, when $x\to 0^-$, then so does $\varphi(x)$.

  \medskip

  (2) Using the definition of $\psi_{-1}$, and the fact that the
  function $\sinh$ is increasing and by \eqref{psi-1_2} we have
  \begin{align*}%\label{psi-1_3}
    &e^{-\psi_{-1}(a,x) + \frac{\log(-2x)}{1-a}} =
    e^{-\psi_{-1}(x)}(-2x)^{\frac{1}{1-a}} =
    \frac{(-2x)^{\frac{1}{1-a}} \sinh(a\psi_{-1}(a,x))}{x}\\ & \leq
    \frac{(-2x)^{\frac{1}{1-a}}
      \sinh\left(a\frac{\log(-2x)}{1-a}\right)}{x} =
    \frac{(-2x)^{\frac{1+a}{1-a}} - (-2x)^{\frac{1-a}{1-a}}}{2x} = 1 -
    (-2x)^{\frac{2a}{1-a}}.
  \end{align*}

  \medskip

  (3) Let $\varphi(x) = \psi_{-1}(x)-\frac{\log(-2x)}{1-a}$. This
  proof follows closely that of Theorem~\ref{thmpsi0}, but instead of
  relation~\eqref{thmpsi0_2} we use \eqref{est-13}, and
  Lemma~\ref{thmpsi0_lem} is replaced by Lemma~\ref{thmpsi1_lem}.
\end{proof}

Proceeding in a similar manner as in Section~\ref{section9}, we get an
expansion of $\psi_{-1}$ in terms of $Z=(-2x)^{\frac {2a}{1-a}}$:
\begin{multline*}
  \psi_{-1}(a,x) \thicksim \frac{\log(-2x)}{1-a} +
  \frac{1}{1-a}(-2x)^{\frac{2a}{1-a}} +
  \frac{1+3a}{2(1-a)^2}(-2x)^{\frac{4a}{1-a}} + \\
  \frac{10a^2+7a+1}{3(1-a)^3}(-2x)^{\frac{6a}{1-a}}.
\end{multline*}

\end{document}